\documentclass[a4paper,UKenglish]{lipics-v2016}

\usepackage{microtype}
\usepackage{amsmath,amssymb,amsfonts,mathrsfs}
\usepackage[all,2cell]{xy}
\usepackage{graphicx}
\usepackage{hyperref}
\usepackage{anyfontsize}

\usepackage{etoolbox}
\makeatletter
\setbool{@fleqn}{false}
\makeatother

\newcommand{\tr}[1]{\xrightarrow{#1}}
\newcommand{\tl}[1]{\xleftarrow{#1}}
\newcommand{\maps}{{\colon}}

\def \catC {\mathcal{C}}
\def \subc {\mathcal{A}}
\newcommand{\op}[1]{{#1}^{\scriptscriptstyle op}}
\newcommand{\Span}[1]{\mathsf{Span}(#1)}
\newcommand{\Spanc}[1]{\mathsf{Span}(#1)}
\newcommand{\Cospan}[1]{\mathsf{Cospan}(#1)}
\newcommand{\Rel}[1]{\mathsf{Rel}(#1)}
\newcommand{\Corel}[1]{\mathsf{Corel}(#1)}
\def \PROP {\mathbf{Prop}} 
\def \CAT {\mathbf{Cat}} 
\def \SET {\mathsf{Set}} 
\def \poi {\,\ensuremath{;}\,} 
\def \tns {\ensuremath{\oplus}} 
\def \id {\mathit{id}} 

\newcommand\pair[1]{\langle #1 \rangle} 
\newcommand\copair[1]{[#1]} 
\newcommand\ord[1]{\overline{#1}}

\def \CospanToCorel {\Gamma}
\def \SpanToCorel {\Pi}
\newcommand \+[1] {+_{\scriptscriptstyle \lvert #1 \rvert}}

\newcommand{\fact}[2]{({#1} , {#2})}
\newcommand{\Cepi}{\mathcal{E}} 
\newcommand{\Cmono}{\mathcal{M}} 

\newcommand{\Triv}{\mathbf{1}} 
\newcommand{\F}{\mathsf{F}} 
\newcommand{\Inj}{\mathsf{In}}  
\newcommand{\Surj}{\mathsf{Su}} 
\newcommand{\ER}{\mathsf{ER}}  
\newcommand{\PER}{\mathsf{PER}}  
\def \PF {\mathsf{PF}} 
\newcommand \Vect[1] {\mathsf{Vect}_{\scriptscriptstyle #1}}
\newcommand \SV[1] {\mathsf{SV}_{\scriptscriptstyle #1}}
\def \PID {\mathsf{R}}
\def \field {\mathsf{k}} 
\def \frPID {\mathsf{k}} 
\def \Z {\mathsf{Z}}

\newcommand{\fmod}[1]{\mathsf{FMod}_{\scriptscriptstyle #1}}
\newcommand{\mfmod}[1]{\mathsf{MFMod}_{\scriptscriptstyle #1}} 

\newcommand{\pullbacktop}[4]{%
{#1} \ar@/_/[ddr]_{#4} \ar@/^/[drr]^{#2}%
\ar@{.>}[dr]|-{#3} \\}

\newcommand{\pullbackcorner}[1][dr]{\save*!/#1+1.2pc/#1:(1,-1)@^{|-}\restore}

\newcommand{\cgr}[2][scale=0.45]{\raisebox{0.1em}{\begingroup
\setbox0=\hbox{\includegraphics[#1]{graffles/#2}}%
\parbox{\wd0}{\box0}\endgroup}}

\newcommand\Bmult{\lower5pt\hbox{$\includegraphics[width=20pt]{graffles/Bmult.pdf}$}}
\newcommand\Bcomult{\lower5pt\hbox{$\includegraphics[width=20pt]{graffles/Bcomult.pdf}$}}
\newcommand\Bunit{\cgr[height=10pt]{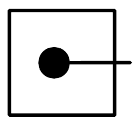}}
\newcommand\Bcounit{\cgr[height=10pt]{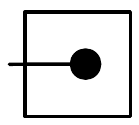}}
\newcommand\EmptyDiag{\cgr[height=10pt]{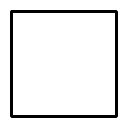}}
\newcommand\IdDiag{\cgr[height=12pt]{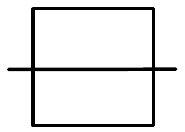}}
\newcommand\scalar{\cgr[height=12pt]{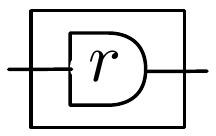}}
\newcommand\scalarop{\cgr[height=12pt]{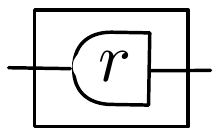}}

\theoremstyle{plain}
\newtheorem{prop}[]{Proposition}
\newtheorem{thm}[]{Theorem}
\newtheorem{lem}[]{Lemma}
\newtheorem{cor}[]{Corollary}

\theoremstyle{definition}
\newtheorem{ass}[]{Assumption}
\newtheorem{defn}[]{Definition}

\newtheorem{rmk}[]{Remark}

\title{A Universal Construction for (Co)Relations\footnote{
Brendan Fong acknowledges support from the Queen Elizabeth Scholarship,
Oxford, and the Basic Research Office of the ASDR\&E through ONR
N00014-16-1-2010.}
}

\author[1]{Brendan Fong}
\author[2]{Fabio Zanasi}
\affil[1]{University of Pennsylvania, United States of America}
\affil[2]{University College London, United Kingdom}

\authorrunning{B. Fong and F. Zanasi} 

\Copyright{Brendan Fong and Fabio Zanasi}


\EventEditors{Filippo Bonchi and Barbara K\"onig}
\EventNoEds{2}
\EventLongTitle{7th Conference on Algebra and Coalgebra in Computer
Science (CALCO 2017)}
\EventShortTitle{CALCO 2017}
\EventAcronym{CALCO}
\EventYear{2017}
\EventDate{June 12--16, 2017}
\EventLocation{Ljubljana, Slovenia}
\EventLogo{}
\SeriesVolume{72}
\ArticleNo{12}

\begin{document}

\maketitle

\begin{abstract}
Calculi of string diagrams are increasingly used to present the syntax and
algebraic structure of various families of circuits, including signal flow
graphs, electrical circuits and quantum processes. In many such approaches, the
semantic interpretation for diagrams is given in terms of relations or
corelations (generalised equivalence relations) of some kind. In this paper we
show how semantic categories of both relations and corelations can be
characterised as colimits of simpler categories. This modular perspective is
important as it simplifies the task of giving a complete axiomatisation for
semantic equivalence of string diagrams. Moreover, our general result unifies
various theorems that are independently found in literature and are relevant for program semantics, quantum computation and control theory.
\end{abstract}

\section{Introduction}

Network-style diagrammatic languages appear in diverse fields as a tool to reason about computational models of various kinds, including signal processing circuits, quantum processes, Bayesian networks and Petri nets, amongst many others. In the last few years, there have been more and more contributions towards a uniform, formal theory of these languages which borrows from the well-established methods of programming language semantics. A significant insight stemming from many such approaches is that a \emph{compositional} analysis of network diagrams, enabling their reduction to elementary components, is more effective when system behaviour is thought as a \emph{relation} instead of a function. 

A paradigmatic case is the one of signal flow graphs, a foundational structure
in control theory: a series of recent
works~\cite{Bonchi2014b,BaezErbele-CategoriesInControl,Bonchi2015,BonchiSZ17,Fong2015}
gives this graphical language a syntax and a semantics where each signal flow
diagram is interpreted as a subspace (a.k.a. linear relation) over streams. The
highlight of this approach is a sound and complete axiomatisation for semantic
equivalence: what is of interest for us is how this result is achieved
in~\cite{Bonchi2014b}, namely through a \emph{modular} account of the domain of
subspaces. The construction can be studied for any field $\field$: one
considers the prop\footnote{A prop is a symmetric monoidal category with
objects the natural numbers \cite{MacLane1965}. It is the typical setting for
studying both the syntax and the semantics of network diagrams.} $\SV{\field}$
whose arrows $n \to m$ are subspaces of $\field^n \times \field^m$, composed as
relations. As shown in in~\cite{interactinghopf,ZanasiThesis}, $\SV{\field}$
enjoys a universal characterisation: it is the pushout (in the category of
props) of props of \emph{spans} and of \emph{cospans} over $\Vect{\field}$, the prop
with arrows $n \to m$ the linear maps $\field^n \to \field^m$:
\begin{equation}\label{intro:linearcube}
\begin{aligned}
    \xymatrix@R=15pt{
    {\Vect{\field} + \op{\Vect{k}}} \ar[r]
    \ar[d]& {\Span{\Vect{\field}}}
    \ar[d] \\
    {\Cospan{\Vect{\field}}} \ar[r] & {\SV{\field}}. \pullbackcorner
    }
  \end{aligned}
\end{equation}
In linear algebraic terms, the two factorisation properties expressed by
\eqref{intro:linearcube} correspond to the representation of a subspace in terms
of a basis (span) and the solution set of a system of linear equations (cospan).
Most importantly, this picture provides a roadmap towards a complete
axiomatisation for $\SV{\field}$: one starts from the domain $\Vect{\field}$ of
linear maps, which is axiomatised by the equations of Hopf algebras, then combines
it with its opposite $\op{\Vect{\field}}$ via two distributive laws of
props~\cite{Lack2004a},  one yielding an axiomatisation for
$\Span{\Vect{\field}}$ and the other one for $\Cospan{\Vect{\field}}$. Finally,
merging these two axiomatisations yields a complete axiomatisation for $\SV{\field}$,
called the theory of interacting Hopf algebras~\cite{interactinghopf,ZanasiThesis}.

\medskip

It was soon realised that this modular construction was of independent interest, and perhaps evidence of a more general phenomenon. In~\cite{Zanasi16} it is shown that a similar construction could be used to characterise the prop $\ER$ of equivalence relations, using as ingredients $\Inj$, the prop of injections, and $\F$, the prop of total functions. The same result is possible by replacing equivalence relations with partial equivalence relations and functions with partial functions, forming a prop $\PF$. In both cases, the universal construction yields a privileged route to a complete axiomatisation, of $\ER$ and of $\PER$ respectively~\cite{Zanasi16}. 
\begin{equation}\label{intro:cartesiancubes}
\begin{aligned}
    \xymatrix@R=15pt{
    {\Inj + \op{\Inj}} \ar[r]
    \ar[d]& {\Span{\Inj}}
    \ar[d] \\
    {\Cospan{\F}} \ar[r] & {\ER} \pullbackcorner
    }
    \qquad \qquad
    \xymatrix@R=15pt{
    {\Inj + \op{\Inj}} \ar[r]
    \ar[d]& {\Span{\Inj}}
    \ar[d] \\
    {\Cospan{\PF}} \ar[r] & {\PER} \pullbackcorner.
    }
  \end{aligned}
\end{equation}
Even though a pattern emerges, it is certainly non-trivial: for instance, if one naively mimics the linear case~\eqref{intro:linearcube} in the attempt of characterising the prop of relations, the construction collapses to the terminal prop $\Triv$.
\begin{equation}\label{intro:cartesiancubecollapses}
\begin{aligned}
    \xymatrix@R=15pt{
    {\F + \op{\F}} \ar[r]
    \ar[d] & {\Span{\F}} \ar[d]^{} \\
    {\Cospan{\F}} \ar[r]_-{} & {\Triv}\pullbackcorner
    }
  \end{aligned}
\end{equation}

More or less at the same time, diagrammatic languages for various families of
circuits, including linear time-invariant dynamical systems \cite{Fong2015},
were analysed using so-called \emph{corelations}, which are generalised
equivalence relations~\cite{Fon16,CF,Fon17,Fong2013}. Even though they were not
originally thought of as arising from a universal construction like the examples above, corelations still follow a modular recipe, as they are expressible as a quotient of $\Cospan{\catC}$, for some prop $\catC$. Thus by analogy we can think of them as yielding one half of the diagram
\begin{equation}\label{intro:corelations}
\begin{aligned}
    \xymatrix@R=15pt{
    {\catC + \op{\catC}}
    \ar[d] &  \\
    {\Cospan{\catC}} \ar[r]_-{} & {\Corel{\catC}}.
    }
  \end{aligned}
\end{equation}
In this paper we clarify the situation by giving a unifying perspective for all these constructions. We prove a general result, which
\begin{itemize}
\item implies \eqref{intro:linearcube} and \eqref{intro:cartesiancubes} as special cases;
\item explains the failure of \eqref{intro:cartesiancubecollapses};
\item extends \eqref{intro:corelations} to a pushout recipe for corelations.
\end{itemize}
More precisely, our theorem individuates sufficient conditions for characterising the category $\Rel{\catC}$ of $\catC$-relations as a pushout. A dual construction yields the category $\Corel{\catC}$ of $\catC$-corelations as a pushout. For the case of interest when $\catC$ is a prop, the two constructions look as follows.
\begin{equation}
\begin{aligned}
    \xymatrix@R=15pt{
      {\subc + \op{\subc}} \ar[r] \ar[d] & {\Spanc{\catC}}
      \ar[d] \\
    {\Cospan{\subc}} \ar[r] & {\Rel{\catC}} \pullbackcorner.
    }
 \qquad \qquad 
    \xymatrix@R=15pt{
      {\subc + \op{\subc}} \ar[r] \ar[d] & {\Spanc{\subc}}
      \ar[d] \\
    {\Cospan{\catC}} \ar[r] & {\Corel{\catC}} \pullbackcorner.
    }
  \end{aligned}
\end{equation}
The variant ingredient $\subc$ is a subcategory of $\catC$. In order to make the constructions possible, $\subc$ has to satisfy certain requirements in relation with the factorisation system $(\Cepi, \Cmono)$ on $\catC$ which defines $\catC$-relations (as jointly-in-$\Cmono$ spans) and $\catC$-corelations (as jointly-in-$\Cepi$ cospans). For instance, taking $\subc$ to be $\catC$ itself succeeds in \eqref{intro:linearcube} (and in fact, for any abelian $\catC$), but fails in \eqref{intro:cartesiancubecollapses}. 

Besides explaining existing constructions, our result opens the lead for new applications. In particular, we observe that under mild conditions the construction lifts to the category $\catC^T$ of $T$-algebras for a monad $T \colon \catC \to \catC$. We leave the exploration of this and other ramifications for future work.

\paragraph*{Synopsis} 
Section~\ref{sec:corelations} introduces the necessary preliminaries about
factorisation systems and (co)relations, and shows the subtleties of mapping
spans into corelations in a functorial way.  Section~\ref{sec:theorem} states
our main result and some of its consequences.  We first formulate the
construction for categories (Theorem~\ref{thm:corelations}), and then for props
(Theorem~\ref{thm:corelationsPROPs}), which are our prime object of interest in
applications. Section~\ref{sec:examples} is devoted to show various instances of
our construction. We illustrate the case of equivalence relations, of partial
equivalence relations, of subspaces, of linear corelations, and finally of
relations of algebras. Finally, Section~\ref{sec:conclusion} summarises our
contribution and looks forward to further work. An appendix contains the proofs
of Theorems~\ref{thm:corelations} and \ref{thm:corelationsPROPs}. 

\paragraph*{Conventions} 
We write $f \poi g$ for composition of $f \colon X \to Y$ and $g \colon Y \to Z$ in a category $\catC$. It will be sometimes convenient to indicate an arrow $f \colon X \to Y$ of $\catC$ as $X\tr{f \in \catC}Y$ or also $\tr{\in \catC}$, if names are immaterial. Also, we write $X \tl{f \in \catC} Y$ for an arrow $X \tr{f \in \op{\catC}} Y$. We use $\tns$ for the monoidal product in a monoidal category, with unit object $I$. Monoidal categories and functors will be strict when not stated otherwise.

\section{(Co)relations} \label{sec:corelations}
In this section we review the categorical approach to relations, based on the
observation that in $\SET$ they are the jointly mono spans. We introduce in
parallel the dual notion, called corelations \cite{Fon16}: these are jointly epi
cospans and can be seen as an abstraction of the concept of equivalence
relation.

\begin{defn}
  A {\bf factorisation system} $(\Cepi,\Cmono)$ in a category
  $\mathcal C$ comprises subcategories $\Cepi$, $\Cmono$ of $\mathcal
  C$ such that
  \begin{enumerate}
    \item $\Cepi$ and $\Cmono$ contain all isomorphisms of $\mathcal
      C$.
    \item  every morphism $f \in \mathcal C$ admits a factorisation $f=e;m$, $e \in \Cepi$, $m \in \Cmono$.
\item given $f,f'$, with factorisations $f = e;m$, $f' = e';m'$ of the above
  sort, for every $u$, $v$ such that $f;v = u;f'$
  there exists a unique $s$ making the following diagram commute.
    \[
    \xymatrixcolsep{2pc}
    \xymatrixrowsep{2pc}
    \xymatrix{
      \ar[r]^e \ar[d]_u & \ar[r]^m \ar@{-->}[d]^{\exists! s} &  \ar[d]^v \\
       \ar[r]_{e'}& \ar[r]_{m'} &
    }
  \]
  \end{enumerate}
\end{defn}

\begin{defn}
  Given a category $\catC$, we say that a subcategory $\subc$ is {\bf stable
  under pushout} if for every pushout square 
  \[
    \xymatrix{
      \ar[r]^a \ar[d] & \ar[d]\\
      \ar[r]^f & {}
    }
  \]
  such that $a \in \subc$, we also have that $f \in \subc$. Similarly, we say
  that $\subc$ is {\bf stable under pullback} if for every pullback square
  labelled as above $f \in \subc$ implies $a \in \subc$.

  A factorisation system $(\Cepi,\Cmono)$ is {\bf stable} if $\Cepi$ is stable
  under pullback, {\bf costable} if $\Cmono$ is stable under pushout, and {\bf
  bistable} if it is both stable and costable.
\end{defn}

Examples of bistable factorisation systems include the trivial factorisation
systems $(\mathcal I_{\catC},\catC)$ and $(\catC,\mathcal I_{\catC})$ in any category
$\catC$, where $\mathcal I_{\catC}$ is the subcategory containing exactly the
isomorphisms in $\catC$, the epi-mono factorisation system in any topos, or the
epi-mono factorisation system in any abelian category. Stable factorisation
systems include the (regular epi, mono) factorisation system in any regular
category, such as any category monadic over $\SET$. Dually, costable
factorisation systems include the (epi, regular mono) factorisation system in any
coregular category, such as the category of topological spaces and continuous
maps.

\begin{defn}~
\begin{itemize}
\item Given a category $\catC$ with pushouts, the category $\Cospan{\catC}$ has the same objects as $\catC$ and arrows $X \to Y$ isomorphism classes of cospans $X \tr{f}  \tl{g} Y$ in $\catC$. The composite of $X \tr{f}  \tl{g} Y$ and $Y \tr{h}  \tl{i} Z$ is obtained by taking the pushout of $\tl{g} \tr{h}$.
\item Given a category $\catC$ with pullbacks, the category $\Span{\catC}$ has the same objects as $\catC$ and arrows $X \to Y$ isomorphism classes of spans $X \tl{f}  \tr{g} Y$ in $\catC$. The composite of $X \tl{f}  \tr{g} Y$ and $Y \tl{h}  \tr{i} Z$ is obtained by taking the pullback of $\tr{g}\tl{h}$.
\end{itemize}
\end{defn}

When $\catC$ also has a (co)stable factorisation system, we may define a category
of (co)relations with respect to this system.

\begin{defn}~\label{def:corel}
  \begin{itemize}
    \item Given a category $\catC$ with pushouts and a costable factorisation system
   $(\Cepi,\Cmono)$, the category $\Corel{\catC}$ has the same objects as $\catC$.
   The arrows $X \to Y$ are equivalence classes of cospans $X \tr{f} N \tl{g} Y$ under
   the symmetric, transitive closure of the following relation: two cospans $X \tr{f} N
   \tl{g} Y$ and $X \tr{f'} N' \tl{g'} Y$ are related if there exists $N
   \tr{m}
   N'$ in $\Cmono$ such that
  \begin{equation}\label{eq:defcorelations}
    \begin{aligned}
    \xymatrix@R=5pt{
      & N   \ar[dd]^m  \\
      X \ar[ur]^{f} \ar[dr]_{f'}&& Y \ar[ul]_{g}\ar[dl]^{g'}\\
      & N'  
    }
  \end{aligned}
  \end{equation}
  commutes. This notion of equivalence respects composition of cospans, and so
  $\Corel{\catC}$ is indeed a category. We call the morphisms in this category
  {\bf corelations}. 
  
\item Given a category $\catC$ with pullbacks and a stable factorisation
   system, we can dualise the above to define the category
   $\Rel{\catC}$ of {\bf relations}. 
\end{itemize}
  \end{defn}

(Co)stability is needed in order to ensure that composition of (co)relations is
associative, \emph{cf.}~\cite[\S 3.3]{Fon16}. For proofs it is convenient to
give an alternative description of (co)relations.

\begin{prop} \label{lemma:charCospanToCorel}
When $\catC$ has binary coproducts, corelations are in one-to-one
  correspondence with isomorphism classes of cospans such that the
  copairing $\copair{p,q} \maps X + Y \to N$ lies in $\Cepi$. 
  
When $\catC$ has binary products, relations are in one-to-one
  correspondence with isomorphism classes of spans such that the
  pairing $\pair{f,g} \maps N \to X \times Y$ lies in $\Cmono$. 
\end{prop}

We refer to Appendix~\ref{sec:minorproof} for a proof of the proposition.

We call a span $\tl{f}\tr{g}$ {\bf jointly-in-$\Cmono$} if the pairing $\langle
f,g\rangle$ lies in $\Cmono$, and analogously for $\Cepi$ and for cospans. To
each relation there is thus, up to isomorphism, a canonical representation as a
jointly-in-$\Cmono$ span, and similarly to each corelation a jointly-in-$\Cepi$
cospan.

\begin{example} \label{ex.corels}
  Many examples of relations and corelations are already familiar.
  \begin{itemize}
    \item The category $\SET$ is bicomplete and has a bistable epi-mono
      factorisation system. Relations with respect to this factorisation system
      are simply the usual binary relations, while corelations from $X \to Y$ in
      $\SET$ are surjective functions $X+Y \to N$; thus their isomorphism
      classes---the arrows of $\Corel{\SET}$---are partitions, or equivalence relations on $X+Y$. 
    \item The category of vector spaces over a field $\field$ is abelian, and hence bicomplete with a
      bistable epi-mono factorisation system. The categories of
      relations and corelations are isomorphic: a morphism $X \to Y$ in these
      categories can be thought of as a linear relations, i.e. a subspace of $X \times Y$.  
       \item In any category $\catC$ the trivial morphism-isomorphism factorisation
      system $(\catC,\mathcal I_{\catC})$ is bistable. Relations with respect to $(\catC,\mathcal I_{\catC})$ are equivalence classes of isomorphisms $N
      \stackrel\sim\to X \times Y$, and hence there is a unique relation between
      any two objects. Corelations are just cospans.
    \item Dually, relations with respect to the isomorphism-morphism
      factorisation $(\mathcal I_{\catC}, \catC)$ are just spans, and there is
      a unique corelation between any two objects.  
        \end{itemize}
\end{example}

We now study the functorial interpretation of cospans and spans as
corelations. This discussion is instrumental in our universal construction for
corelations (Theorem \ref{thm:corelations}).

First, given two categories with the same collections of objects, we may speak
of {\bf identity-on-objects (ioo)} functors between them, i.e. functors that are
the identity map on objects. Four examples of such functors will become relevant in the next section:
\begin{equation}\label{eq:functorsCToSpanCospan}
\begin{aligned}
\catC \to \Cospan{\catC} \text{ maps } \tr{f} \text{ to } \tr{f}\tl{id} \qquad &
\qquad \phantom{\scriptstyle op}
\catC \to \Span{\catC} \text{ maps } \tr{f} \text{ to } \tl{id}\tr{f}\\
 \op{\catC} \to \Cospan{\catC} \text{ maps }  \tl{g } \text{ to } \tr{id}\tl{g}
 \qquad & \qquad
 \op{\catC} \to \Span{\catC} \text{ maps }  \tl{g } \text{ to } \tl{g}\tr{id}
 \end{aligned}
\end{equation}

We are now ready to discuss the canonical map from cospans to corelations. This
is simple: one just interprets a cospan representative as its corelation
equivalence class.

\begin{defn}~\label{def:cospantocorel}
Let $\catC$ be a category equipped with a costable factorisation system
  $(\Cepi,\Cmono)$. We define $\CospanToCorel \colon \Cospan{\catC} \to
  \Corel{\catC}$ as the ioo functor mapping the isomorphism class of cospans
  represented by $X\tr{f}N\tl{g}Y$ to the corelation represented by this cospan.
\end{defn}

It is straightforward to check that this is well-defined. Moreover, 
\begin{prop}\label{prop:CospanToCorelFull} $\CospanToCorel \colon \Cospan{\catC} \to \Corel{\catC}$ is full. \end{prop}
\begin{proof}
Let $a$ be a corelation. Then choosing some representative $X \to N \leftarrow
Y$ of $a$ gives a cospan whose $\CospanToCorel$-image is $a$.
\end{proof}

Mapping spans to corelations is subtler. Given a span, we may obtain a cospan by
taking its pushout. When $\catC$ has pushouts and pullbacks, this defines a
function on morphisms $\Span{\catC} \to \Cospan{\catC}$. This function is
rarely, however, a functor: it may fail to preserve composition. To turn it into
a functor, two tweaks are needed: first, we restrict to a subcategory
$\Span\subc$ of $\Span{\catC}$, for some carefully chosen subcategory $\subc
\subseteq \catC$, and second, we take the jointly-in-$\Cepi$ part of the
pushout. We call the resulting functor $\SpanToCorel$, as it takes the
\emph{pushout} and then \emph{projects}.

How do we choose $\subc$? Given a cospan $X \to A \leftarrow Y$, we may take its
pullback to obtain a span $X \leftarrow P \to Y$, and then pushout this span in
$\catC$ to obtain a cospan $X \to Q \leftarrow Y$. This gives a diagram
\[
  \xymatrix@R=10pt@C=40pt{
    & X \ar[dr] \ar@/^1em/[drr] \\
    P \ar[ur] \ar[dr] && Q \ar@{-->}[r]^f & A \\
    & Y \ar[ur] \ar@/_1em/[urr]
  }
\]
where the map $f$ exists and is unique by the universal property of the pushout.
We want this map $f$ to lie in $\Cmono$: by Definition~\ref{def:corel}, this implies that $X \to A \leftarrow
Y$ and $X \to Q \leftarrow Y$ represent the same corelation. This condition is
reminiscent of that introduced by Meisen in her work on so-called categories of
pullback spans \cite{Mei74}.

Note that this pullback and pushout take place in $\catC$. We nonetheless
ask $\subc$ to be closed under pullback, so spans $\tl{f \in \subc}\tr{g\in \subc}$ do indeed form a subcategory $\Span\subc$ of $\Span\catC$.\footnote{Calling this subcategory $\Span\subc$ is a slight abuse of notation: it may be
the case that $\subc$ itself has pullbacks, and we have not proved that these
agree with pullbacks in $\catC$. Nonetheless, this conflict does not cause
trouble in any of our examples below, and we stick to this convention for
notational simplicity.}

\begin{prop}~\label{prop:spantocorel}
  Let $\catC$ be a category equipped with a costable factorisation system
  $(\Cepi,\Cmono)$.  Let $\subc$ be a subcategory of $\catC$ containing all
  isomorphisms and stable under pullback.  Further suppose that the canonical
  map given by the pushout of the pullback of a cospan in $\subc$ lies in
  $\Cmono$. Then mapping a span in $\subc$ to the jointly-in-$\Cepi$ part of its
  pushout cospan defines an ioo functor $\Pi\maps \Span{\subc} \to \Corel{\catC}$.
\end{prop}
\begin{proof}
  Recall that $\Span{\subc}$ is generated by morphisms of the form
  $\tl{id}\tr{f \in \subc}$ and $\tl{f \in \subc}\tr{id}$. It is thus enough to
  show $\SpanToCorel$ preserves composition on arrows of these two types. There
  exist four cases: (i) $\tl{id}\tr{f}\tl{id}\tr{g}$, (ii)
  $\tl{f}\tr{id}\tl{g}\tr{id}$, (iii) $\tl{f}\tr{id}\tl{id}\tr{g}$, and (iv)
  $\tl{id}\tr{f}\tl{g}\tr{id}$. The first three cases are straightforward to
  prove, and in fact hold when mapping $\Span{\catC} \to \Cospan{\catC}$. It is
  the case (iv) that needs our restriction to $\Span\subc$. There $\SpanToCorel(\tl{id}\tr{f})\SpanToCorel(\tl{g}\tr{id})$
  is represented by the cospan $\tr{f}\tl{g}$, while
  $\SpanToCorel(\tl{id}\tr{f}\tl{g}\tr{id})$ is the represented by the
  pushout $\tr{p}\tl{q}$ of the pullback of $\tr{f}\tl{g}$. But by hypothesis,
  there exists a unique $\tr{m \in \Cmono}$ making the following diagram commute.
  \[
    \xymatrix@C=30pt@R=10pt{
      &  \ar[dr]^p \ar@/^1em/[drr]^f \\
      &&  \ar[r]^m &  \\
      & \ar[ur]_q \ar@/_1em/[urr]_g
    }
  \]
This implies that $\tr{p}\tl{q}$ and $\tr{f}\tl{g}$ represent the
  same corelation, and so $\SpanToCorel$ is functorial.
\end{proof}

For example, if the category $\Cmono$ has pullbacks and these coincide with pullbacks in
$\catC$, then we can take $\subc=\Cmono$. If $\catC$ is abelian, we can take $\subc=\catC$.

\section{Main theorem: a universal property for (co)relations}\label{sec:theorem}

This section states our main result and some consequences. We first fix our
ingredients. 
\begin{ass}\label{ass:thcorelations}
  Let $\catC$ be a category with
  \begin{itemize}
    \item pushouts and pullbacks;
  \item a costable factorisation system $\fact{\Cepi}{\Cmono}$ with $\Cmono$ a subcategory of the
  monos in $\catC$;
  \item a subcategory $\subc$ of $\catC$ containing $\Cmono$, stable under pullback, and
  such that the canonical map given by the pushout of the pullback of a
  cospan in $\subc$ lies in $\Cmono$.
  \end{itemize}
\end{ass}

Building on the results of Section~\ref{sec:corelations}, the second requirement
above allows us to form a category $\Corel{\catC}$ of corelations, whereas the
third yields a functor $\SpanToCorel \colon \Span{\subc} \to \Corel{\catC}$.  We
shall also use the functor $\CospanToCorel \colon \Cospan{\catC} \to
\Corel{\catC}$ (Definition~\ref{def:cospantocorel}) and a category
$\subc\+{\subc}\op{\subc}$: its objects are those of $\subc$ and the morphisms
$X \to Y$ are `zigzags' $X\tr{f}\tl{g}\tr{h} \dots \tl{k} Y$ in $\subc$. There
are ioo functors from $\subc\+{\subc}\op{\subc}$ to $\Cospan{\catC}$ and to
$\Span{\catC}$, defined on morphisms by taking colimits, respectively limits of
zigzags---equivalently, they are defined by pointwise application of the
functors in~\eqref{eq:functorsCToSpanCospan}.
\footnote{More abstractly, one can
  see $\subc\+{\subc}\op{\subc}$ as the pushout of $\subc$ and $\op{\subc}$ over
  the respective inclusions of $\lvert\subc\rvert$, the discrete category on the
  objects of $\subc$. The functors $\Span{\subc} \tl{} \subc\+{\subc}\op{\subc}
  \tr{} \Cospan{\catC}$ are then those given by the universal property with
  respect to (suitable restrictions of) the functors in
  \eqref{eq:functorsCToSpanCospan}.} 
We make all these components interact in our main theorem.

\begin{thm}
 \label{thm:corelations}
   Let $\catC$ and $\subc$ be as in Assumption \ref{ass:thcorelations}. Then the
   following is a pushout in $\CAT$:
\begin{equation}\tag{$\star$}
\label{eq:pushoutCorel}
\begin{aligned}
    \xymatrix@C=40pt{
      {\subc \+{\subc} \op{\subc}} \ar[r] \ar[d] & {\Spanc{\subc}}
      \ar[d]^{\SpanToCorel} \\
    {\Cospan{\catC}} \ar[r]_-{\CospanToCorel} & {\Corel{\catC}} \pullbackcorner
    }
  \end{aligned}
\end{equation}
\end{thm}
We leave a complete proof of this theorem to Appendix~\ref{sec:proof}. In a
nutshell, the key point is that, in light of~\eqref{eq:defcorelations},
$\Corel{\catC}$ differs from $\Cospan{\catC}$ precisely because it has the extra
equations $\tr{m}\tl{m} = \tr{\id}\tl{\id}$, with $\tr{m} \in \Cmono$. But these
equations arise by pullback squares in $\subc$, and so are equations of zigzags
in $\Span{\subc}$ (\emph{cf.}\ \eqref{eq:quotientPb} below). Moreover, the
remaining equations of $\Span{\subc}$ can be generated by using these together
with a subset of the equations of $\Cospan{\catC}$. Hence adding the equations
of $\Span{\subc}$ to those of $\Cospan{\catC}$ gives precisely $\Corel{\catC}$,
and we have a pushout square.

\medskip

We now discuss some observations, consequences and examples.

\begin{rmk}\label{rmk:Msuffices}
If any such $\subc$ exists, then we may always take $\subc=\Cmono$ and
the theorem holds. We record the above, more general, theorem as it explains
preliminary results in this direction already in the literature; see the abelian
case and examples for details.
\end{rmk}

Next, we formulate the dual version of the theorem, which yields a characterisation for relations. It is based on a dual version of Assumption \ref{ass:thcorelations}.

\begin{ass}\label{ass:threlations}
  Let $\catC$ be a category with
  \begin{itemize}
    \item pushouts and pullbacks;
  \item a stable factorisation system $\fact{\Cepi}{\Cmono}$ with $\Cepi$ a subcategory of the
  epis in $\catC$; 
  \item a subcategory $\subc$ of $\catC$ containing $\Cepi$, stable under pushout, and
  such that the canonical map given by the pullback of the pushout of a
  span in $\subc$ lies in $\Cepi$.
  \end{itemize}
\end{ass}

\begin{cor}[Dual case] \label{cor.dual}
  Let $\catC$ and
  $\subc$ be as in Assumption \ref{ass:threlations}. Then the following is a pushout square in $\CAT$.
  \begin{equation}\tag{$\circ$}\label{eq:pushoutRel}
    \begin{aligned}
      \xymatrix@C=40pt{
	{\subc \+{\subc} \op{\subc}} \ar[r] \ar[d] & {\Cospan{\subc}}
	\ar[d]^{} \\
	{\Span{\catC}} \ar[r]_-{} & {\Rel{\catC}} \pullbackcorner
      }
    \end{aligned}
  \end{equation}
\end{cor}
\begin{proof}
  This corollary is obtained by noting that, given a stable factorisation system
  $(\Cepi,\Cmono)$ in $\catC$, with $\Cepi$ a subcategory of the epis, we have a
  costable factorisation system $(\op{\Cmono},\op{\Cepi})$ in $\op{\catC}$, with
  $\op\Cepi$ a subcategory of the monos. Proposition \ref{prop:spantocorel} then
  gives a functor $\Cospan{\subc}= \Span{\op\subc} \longrightarrow \Rel{\catC} =
  \Corel{\op\catC}$. Noting also that $\subc \+{\subc} \op{\subc} = \op\subc
  \+{\subc} \op{(\op\subc)}$ and $\Span{\catC} = \Cospan{\op\catC}$, we can
  hence apply Theorem~\ref{thm:corelations}.
\end{proof}

As a notable instance of Theorem~\ref{thm:corelations}, we can specialise to the case of abelian
categories and their epi-mono factorisation system. In this case we can simply
pick $\subc$ to be $\catC$ itself.

\begin{cor}[Abelian case]
 \label{thm:pushoutAbelian}
  Let $\catC$ be an abelian category. Then the following is a pushout square in
  $\CAT$:
\begin{equation}\tag{$\triangle$}
\label{eq:pushoutAbelian}
\begin{aligned}
    \xymatrix@C=40pt{
    {\catC \+{\catC} \op{\catC}} \ar[r]
    \ar[d] & {\Span{\catC}}
    \ar[d]^{\SpanToCorel} \\
    {\Cospan{\catC}} \ar[r]_-{\CospanToCorel} & {\Corel{\catC}} \cong \Rel{\catC} \pullbackcorner
    }
  \end{aligned}
\end{equation}
  where we take (co)relations with respect the epi-mono factorisation system. 
\end{cor}
\begin{proof}
  As $\catC$ is abelian, it is finitely bicomplete and has a bistable
  factorisation system given by epis and monos.  Furthermore, we need not
  restrict our spans to some subcategory $\subc$: in an abelian category the
  pullback of a cospan $X\tr{f} A\tl{g}Y$ can be computed via the kernel of the
  joint map $X \oplus Y\tr{[f,-g]}A$, and similarly pushouts can be computed via
  cokernel, whence the canonical map from the pushout of the pullback of a given
  cospan to itself is always mono, being the inclusion of the image of the joint
  map into the apex.  Similarly, the map from a span to the pullback of its
  pushout is simply the joint map with codomain restricted to its image, and
  hence always epi.  Thus $\catC$ meets both Assumptions \ref{ass:thcorelations}
  and \ref{ass:threlations} with $\subc = \catC$.  Then the category of
  corelations is the pushout of the span $\Cospan{\catC} \leftarrow \catC
  \+{\catC} \op{\catC} \to \Span{\catC}$. But by the dual theorem (Corollary
  \ref{cor.dual}), the pushout of this span is also the category of relations.
  Thus the two categories are isomorphic.  Explicitly, the isomorphism is given
  by taking a corelation to the jointly mono part of its pullback span, and
  taking a relation to the jointly epi part of its pushout cospan.
\end{proof}

\begin{rmk}
  In Theorem~\ref{thm:corelations}, the diagram $\Cospan{\catC} \leftarrow
\subc \+{\subc} \subc \to \Span{\subc}$ `knows' only about $\Cmono$, not the
factorisation system $(\Cepi,\Cmono)$. This is enough, however, since if $\Cmono$
is part of a factorisation system, then the factorisation system is unique.

  Indeed, suppose we have $\Cepi$, $\Cepi'$ such that both $(\Cepi,\Cmono)$ and
  $(\Cepi',\Cmono)$ are factorisation systems. Take $e \in \Cepi$. Then the
  factorisation system $(\Cepi',\Cmono)$ gives a factorisation $e = e';m_1$,
  while $(\Cepi,\Cmono)$ gives a factorisation $e' = e_2;m_2$.
  By substitution, we have $e = e_2;m_2;m_1$. By uniqueness of factorisation, we can then assume without loss of
  generality that $e = e_2$ and $m_2;m_1 = id$. Next, using $e=e_2$ and
  substitution in $e' = e_2;m_2$, we similarly arrive at $m_1;m_2 = id$. Thus $m_1$ is
  an isomorphism, and hence lies in $\Cepi'$. This implies that $e = e';
  m_1 \in \Cepi'$, and hence $\Cepi \subseteq \Cepi'$. We may similarly show
  that $\Cepi' \subseteq \Cepi$, and hence that the two categories are equal.
\end{rmk}


The next corollary is instrumental in giving categories of (co)relations a presentation by generators and equations.

\begin{cor}\label{cor:presentation} Suppose $\subc$ and $\catC$ are as in Assumption \ref{ass:thcorelations}. Then $\Corel{\catC}$ is freely generated by the objects of $\catC$ and arrows $\tr{f}$, $\tl{g}$ of $\catC$ quotiented by 
\begin{align}
 \tr{f \in \subc}\tl{g \in \subc} \ = \ \tl{p \in \subc}\tr{q \in \subc} && \text{whenever }\tl{p}\tr{q} \text{ pulls back } \tr{f}\tl{g} \label{eq:quotientPb} \\
 \tl{f \in \catC}\tr{g \in \catC} \ = \ \tr{p \in \catC}\tl{q \in \catC} && \text{whenever }\tr{p}\tl{q} \text{ pushes out } \tl{f}\tr{g}. \label{eq:quotientPo}
\end{align}
Equivalently, $\Corel{\catC}$ is the quotient of $\Cospan{\catC}$ by~\eqref{eq:quotientPb}. A dual statement holds for $\Rel{\catC}$.
\end{cor}

Note that, in light of Remark \ref{rmk:Msuffices}, one may also replace \eqref{eq:quotientPb} by the subset of axioms
\begin{align*}
 \tr{f \in \Cmono}\tl{g \in \Cmono} \ = \ \tl{p \in \Cmono}\tr{q \in \Cmono} && \text{whenever }\tl{p}\tr{q} \text{ pulls back } \tr{f}\tl{g}.
\end{align*}
As $\Cmono \subseteq \subc$ by Assumption~\ref{ass:thcorelations}, this may give a smaller presentation.

The importance of the above observation stems from the fact that sets of
equations \eqref{eq:quotientPb} and \eqref{eq:quotientPo} yield a presentation
for categories of spans and cospans over $\catC$ respectively. In various
interesting cases, they enjoy a \emph{finitary} axiomatisation, which can be
elegantly described in terms of distributive laws between categories
\cite{RosebrRWood_fact,Lack2004a}. Under this light, Corollary
\ref{cor:presentation} provides a recipe for axiomatising categories of
(co)relations starting from existing results about spans and cospans. For
instance, this is the strategy adopted in the literature to axiomatise finite
equivalence relations~\cite{Zanasi16,CF}, finite partial equivalence relations
\cite{Zanasi16} and finitely-dimensional subspaces \cite{interactinghopf}. All
these are examples of corelations and are treated in Section \ref{sec:examples}
below. 

\paragraph*{The case of props.} As mentioned in the introduction, the motivating examples for our construction are categories providing a semantic interpretation for circuit diagrams. These are typically props (\textbf{pro}duct and \textbf{p}ermutation categories~\cite{MacLane1965}): it is thus useful to phrase our construction in this setting. 

Recall that a prop is a symmetric monoidal category with objects the
natural numbers, in which $n \tns m = n+m$. Props form a category $\PROP$ with
morphisms the ioo strict symmetric monoidal functors. A simplification to Theorem \ref{thm:corelations} is that the
coproduct $\catC + \catC'$ in $\PROP$ is computed as $\catC \+{\catC} \catC'$ in
$\CAT$, because the set of objects is fixed for any prop.

For monoidal structure on $\catC$ to extend to the categories of (co)spans and
(co)relations, it is crucial that the monoidal product respects the ambient
structure.
  
Let $(\catC,\tns)$ be a prop with pushouts, and let $(\subc,\tns)$ be a
  sub-prop. We say that {\bf the monoidal product preserves pushouts} in
  $\subc$ if, for all spans $N \leftarrow Y \to M$ and $N' \leftarrow Y' \to M'$
  in $\subc$, we have an isomorphism
  \[
    (N \tns N') +_{Y \tns Y'} (M \tns M') \cong (N +_Y M) \tns (N' +_{Y'} M').
  \]
  Note that this pushout is taken in $\catC$. This condition holds, for example,
  whenever $\catC$ is monoidally closed. We say the monoidal product preserves
  pullbacks if the analogous condition holds for pullbacks.

  Furthermore, we say that a subcategory $\subc$ is {\bf closed under $\tns$}
  if, given morphisms $f,g$ in $\subc$, the morphism $f\oplus g$ is also in
  $\subc$.

\begin{thm}\label{thm:corelationsPROPs}
  Let $\catC$ and $\subc$ be props satisfying Assumption
  \ref{ass:thcorelations}. Suppose that the monoidal product of $\catC$
  preserves pushouts in $\catC$ and pullbacks in $\subc$, and that $\Cmono$
  is closed under the monoidal product. Then we have a pushout square in $\PROP$
\begin{equation}
\label{eq:pushoutCorelProps}
\begin{aligned}
    \xymatrix@C=40pt{
      {\subc + \op{\subc}} \ar[r] \ar[d] & {\Spanc{\subc}}
      \ar[d]^{\SpanToCorel} \\
    {\Cospan{\catC}} \ar[r]_-{\CospanToCorel} & {\Corel{\catC}} \pullbackcorner.
    }
  \end{aligned}
\end{equation}
\end{thm}

We also state the prop version of the abelian case. An abelian prop is just a
prop which is also an abelian category and where the monoidal product is the
biproduct.

\begin{cor}
 \label{thm:corelationsAbPROP}
  Suppose that $\catC$ is an abelian prop. The following is a pushout in $\PROP$.
\begin{equation}\tag{$\triangle$}
\begin{aligned}
    \xymatrix@C=40pt{
    {\catC + \op{\catC}} \ar[r]
    \ar[d] & {\Span{\catC}}
    \ar[d]^{\SpanToCorel} \\
    {\Cospan{\catC}} \ar[r]_-{\CospanToCorel} & {\Corel{\catC}} \cong \Rel{\catC} \pullbackcorner
    }
  \end{aligned}
\end{equation}
\end{cor}

We leave the proofs of these results to Appendix~\ref{app:spansPROP}.

\section{Examples}\label{sec:examples}

\subsection{From Injections to Equivalence Relations} \label{ex:er}

Our first example concerns the construction of equivalence relations starting
from injective functions. For $n \in \mathbb{N}$, write $\ord{n}$ for the set
$\{0,1,\dots,n\}$, and $\uplus$ for the disjoint union of sets. We fix a prop
$\ER$ whose arrows $n \to m$ are the equivalence relations on
$\ord{n}\uplus\ord{m}$. For composition $e_1 \poi e_2 \maps n \to m$ of
equivalence relations $e_1 \maps n \to z$ and $e_2 \maps z \to m$, one first
defines an equivalence relation on $\ord{n}\uplus\ord{z}\uplus\ord{m}$ by gluing
together equivalence classes of $e_1$ and $e_2$ along common witnesses in
$\ord{z}$, then obtains $e_1 \poi e_2$ by restricting to elements of
$\ord{n}\uplus\ord{m}$. 
 
Equivalence relations are equivalently described as corelations of functions. For this, let $\F$ be the prop whose arrows $n \to m$ are functions from $\ord{n}$ to $\ord{m}$. $\F$ has the usual factorisation system $(\Surj,\Inj)$ given by epi-mono factorisation, where $\Surj$ and $\Inj$ are the sub-props of surjective and of injective functions respectively. Given these data, one can check that $\ER$ is isomorphic to $\Corel{\F}$, the prop of corelations on $\F$.

We are now in position to apply our construction of
Theorem~\ref{thm:corelationsPROPs}. First, we verify
Assumption~\ref{ass:thcorelations} with $\catC$ instantiated as $\F$ and $\subc$
as $\Inj$. The only point requiring some work is the third, which goes as
follows: given a cospan of monos $X \to P \leftarrow Y$, consider $X$, $Y$ as
subsets of $P$. Then the pullback-pushout diagram looks like
\[
  \xymatrix@R=5pt@C=20pt{
    & X \ar[dr] \ar@/^/[drr] \\
    X\cap Y \ar[ur] \ar[dr] && X\cup Y \ar@{-->}[r] & P \\
    & Y \ar[ur] \ar@/_/[urr]
  }
\]
and $X\cup Y \to P$ is the inclusion map, hence a mono in $\Inj$. Therefore, we
can construct the pushout diagram~(\ref{eq:pushoutCorelProps}) as follows:
\begin{equation}\
\label{eq:pushoutER}
\begin{aligned}
    \xymatrix@R=15pt{
    {\Inj + \op{\Inj}} \ar[r]
    \ar[d]& {\Span{\Inj}}
    \ar[d] \\
    {\Cospan{\F}} \ar[r] & {\ER} \pullbackcorner
    }
  \end{aligned}
\end{equation}
This modular reconstruction easily yields a presentation by generators and relations for (the arrows of) $\ER$. Following the recipe of Corollary \ref{cor:presentation}, $\ER$ is the quotient of $\Cospan{\F}$ by all the equations generated by pullbacks in $\Inj$, as in~\eqref{eq:quotientPb}. Now, recall that $\Inj$ is presented (in string diagram notation~\cite{Selinger2009}) by the generator $\Bunit \colon 0 \to 1$, and no equations. Thus, in order to present all the equations of shape~\eqref{eq:quotientPb} it suffices to consider a single pullback square in $\Inj$:
\begin{equation}\label{eq:spaninj}
\vcenter{\xymatrix@=10pt{
& 1 & \\
0\ar[ur]^{\Bunit} && 0 \ar[ul]_{\Bunit} \\
& \ar[ul]^{\EmptyDiag} 0 \ar[ur]_{\EmptyDiag} &
}} \quad \text{, yielding the equation } \Bunit \poi \Bcounit \ = \ \EmptyDiag \poi \EmptyDiag.
\end{equation}
On the other hand, we know $\Cospan{\F}$ is presented by the theory of special
commutative Frobenius monoids (also termed separable Frobenius algebras),
see~\cite{Lack2004a}. Therefore $\ER$ is presented by the generators and
equations of special commutative Frobenius monoids, with the addition
of~\eqref{eq:spaninj}. This is known as the theory of extraspecial commutative
Frobenius monoids~\cite{CF}. This result also appears in~\cite{Zanasi16}, in both cases
without the realisation that it stems from a more general construction.

\subsection{From Functions to the Terminal Prop}\label{ex:triv}

 It is instructive to see a non-example, to show that the assumptions on $\subc$ are not redundant. One may want consider an obvious variation of~\eqref{eq:pushoutER}, where instead of $\Inj$ one takes the whole $\F$ as $\subc$. However, with this tweak the construction collapses: the pushout is the terminal prop $\Triv$ with exactly one arrow between any two objects. 
\begin{equation}
\begin{aligned}
    \xymatrix@R=15pt{
    {\F + \op{\F}} \ar[r]
    \ar[d] & {\Span{\F}} \ar[d]^{} \\
    {\Cospan{\F}} \ar[r]_-{} & {\Triv} \pullbackcorner
    }
  \end{aligned}
\end{equation}
This phenomenon was noted before (\cite{interactinghopf}, see also~\cite[Th.
5.6]{HeunenVicaryBook}), however without an understanding of its relationship
with other (non-collapsing) instances of the same construction.
Theorem~\ref{thm:corelationsPROPs} explains why this case fails where others
succeed: the problem lies in the choice of $\F$ as the subcategory $\subc$.
Indeed, the canonical map given by the pullback of the pushout of any span in
$\subc = \F$ does not necessarily lie in $\Inj$, i.e. it may be not injective. An
example is given by the cospan $0 \to 1 \leftarrow 2$, with the canonical map
from the pushout of the pullback cospan the non-injective map $2 \to 1$:
\[
  \xymatrix@R=5pt@C=35pt{
    & 2 \ar[dr] \ar@/^/[drr] \\
    0 \ar[ur] \ar[dr] && 2 \ar@{-->}[r] & 1. \\
    & 0 \ar[ur] \ar@/_/[urr]
  }
\]

\subsection{From Injections to Partial Equivalence Relations} \label{ex:per}

Partial equivalence relations (PERs) are common structures in program semantics, which date back to the seminal work of Scott~\cite{Scott1975-datatypeslattices} and recently revamped in the study of quantum computations (e.g., \cite{Jacobs_quantumPER,HasuoH11_GOIQuantumPER}). Our approach yields a characterisation for the prop $\PER$ whose arrows $n \to m$ are PERs on $\ord{n}\uplus\ord{m}$, with composition as in $\ER$. The ingredients of the construction generalise Example~\ref{ex:er} from total to partial maps. Instead of $\F$ one starts with $\PF$, the prop of partial functions, which has a factorisation system involving the sub-prop of partial surjections and the sub-prop of injections. The resulting prop of $\PF$-corelations is isomorphic to $\PER$. Theorem~\ref{thm:corelationsPROPs} yields the following pushout
\begin{equation}\
\label{eq:pushoutPER}
\begin{aligned}
    \xymatrix@R=15pt{
    {\Inj + \op{\Inj}} \ar[r]
    \ar[d]& {\Span{\Inj}}
    \ar[d] \\
    {\Cospan{\PF}} \ar[r] & {\PER} \pullbackcorner .
    }
  \end{aligned}
\end{equation}
As in Example~\ref{ex:er}, following Corollary \ref{cor:presentation},~\eqref{eq:pushoutPER} reduces the task of axiomatising $\PER$ to the one of axiomatising $\Cospan{\PF}$ and adding the single equation~\eqref{eq:spaninj} from $\Span{\Inj}$. $\Cospan{\PF}$ is presented by ``partial'' special commutative Frobenius monoids, studied in~\cite{Zanasi16}.

\subsection{From Linear Maps to Subspaces} \label{ex:sv}

We now consider an example for the abelian case: the prop $\SV{\field}$ whose arrows $n \to m$ are $\field$-linear subspaces of $\field^n \times \field^m$, for a field $\field$. Composition in $\SV{\field}$ is relational: $V \poi W = \{(v,w) \mid \exists u .(v,u) \in V, (u,w) \in W\}$. Interest in $\SV{\field}$ is motivated by various recent applications. We mention the case where $\field$ is the field of Laurent series, in which $\SV{\field}$ constitutes a denotational semantics for signal flow graphs~\cite{Bonchi2014b,BaezErbele-CategoriesInControl,Bonchi2015,BonchiSZ17}, and the case $\field = \Z_2$, in which $\SV{\field}$ is isomorphic to the phase-free ZX-calculus, an algebra for quantum observables~\cite{CoeckeDuncanZX2011,BialgAreFrob14}.

Now, in order to apply our construction, note that $\SV{\field}$ is isomorphic to $\Rel{\Vect{k}}$, where $\Vect{k}$ is the abelian prop whose arrows $n \to m$ are the linear maps of type $\field^n \to \field^m$ (the monoidal product is by direct sum). This follows from the observation that subspaces of $\field^n \times \field^m$ of dimension $z$ correspond to mono linear maps from $\field^z$ to $\field^n \times \field^m$, whence to jointly mono spans $n \tl{} z \tr{}m$ in $\Vect{k}$.

We are then in position to use Corollary \ref{thm:corelationsAbPROP}, which yields the following pushout characterisation for $\SV{\field}$.
\begin{equation}\
\label{eq:pushoutSV}
\begin{aligned}
    \xymatrix@R=15pt{
    {\Vect{k} + \op{\Vect{k}}} \ar[r]
    \ar[d]& {\Span{\Vect{k}}}
    \ar[d] \\
    {\Cospan{\Vect{k}}} \ar[r] & {\SV{\field}} \pullbackcorner .
    }
  \end{aligned}
\end{equation}
This very same pushout has been studied in~\cite{BialgAreFrob14} for the $\field = \Z_2$ case. As before, the modular reconstruction suggests a presentation by generators and relations for $\SV{\field}$, in terms of the theories for spans and cospans in $\Vect{k}$. The axiomatisation of $\SV{\field}$ is called the theory of interacting Hopf algebras~\cite{interactinghopf,ZanasiThesis} , as it features two Hopf algebras structures and axioms expressing their combination.

On the top of existing results on $\SV{\field}$, our Corollary \ref{thm:corelationsAbPROP} suggests a
novel perspective, namely that $\SV{\field}$ can be also thought as the prop of
\emph{corelations} over $\Vect{\field}$. This representation can be understood
by recalling the 1-1 correspondence between subspaces of $\field^n \times
\field^m$ and (solution sets of) homogeneous systems of equations $M v = 0$,
where $M$ is a $z \times (n+m)$ matrix. Writing the block decomposition $M =
(M_1\:\vert\: -M_2)$, where $M_1$ is a $z \times n$ matrix and $M_2$ a $z \times
m$ matrix, this is the same as solutions to $M_1v_1 = M_2v_2$. These systems
then yield jointly epi cospans $n \tr{M_1} z \tl{M_2} m$ in $\Vect{\field}$,
that is, corelations.

\subsection{From Free Module Homomorphisms to Linear Corelations} \label{ex:PID}

We now consider the generalisation of the linear case from fields to principal
ideal domains (PIDs). In order to form a prop, we need to restrict our attention
to finitely-dimensional \emph{free} modules over a PID $\PID$. The symmetric
monoidal category of such modules and module homomorphisms, with monoidal
product by direct sum, is equivalent to the prop $\fmod\PID$ whose arrows $n \to
m$ are $\PID$-module homomorphisms $\PID^n \to \PID^m$ or, equivalently,
$m\times n$-matrices in $\PID$. Because of the restriction to free modules,
$\fmod\PID$ is not abelian. However, it is still finitely bicomplete and has a
costable (epi, split mono)-factorisation system.\footnote{The factorisation
given by (epi, mono) morphisms is not unique up to isomorphism, whence the
restriction to split monos---see~\cite{Fong2015}.} Note that the fact that the ring $\PID$ is a PID matters for the existence of pullbacks, as it is necessary for
submodules of free $\PID$-modules to be free---pushouts exist by self-duality
of $\fmod\PID$.

Write $\mfmod\PID$ for the prop of split monos in $\fmod\PID$. It is a
classical, although nontrivial, theorem in control theory that this category
obeys the required condition on pushouts of pullbacks \cite{Fong2015}. Hence
Theorem~\ref{thm:corelationsPROPs} yields the pushout square
\begin{equation}
  \label{eq:pushoutCorelMod1}
  \begin{aligned}
    \xymatrix@C=40pt{
      \fmod\PID + \op{\fmod \PID} \ar[d] \ar[r] & \Span{\mfmod\PID} \ar[d], \\
      \Cospan{\fmod \PID} \ar[r] &\Corel{\fmod\PID} \pullbackcorner
    }
  \end{aligned}
\end{equation}
in $\PROP$. This modular account of $\Corel{\fmod\PID}$ is relevant for the semantics of dynamical systems. When $\PID = \mathbb{R}[s,s^{-1}]$, the ring of Laurent polynomials
in some formal symbol $s$ with coefficients in the reals, the prop
$\Corel{\fmod{\mathbb{R}[s,s^{-1}]}}$ models complete linear time-invariant
discrete-time dynamical systems in $\mathbb{R}$; more details can be found
in~\cite{Fong2015}. In that paper, it is also proven that $\Corel{\fmod\PID}$ is axiomatised by the presentation of $\Cospan{\fmod \PID}$ with the addition of the law $\Bunit\poi\Bcounit \ = \ \EmptyDiag$. By Corollary~\ref{cor:presentation}, it follows that $\Bunit\poi\Bcounit \ = \ \EmptyDiag$ originates by a pullback in $\Span{\mfmod\PID}$ and in this case it is the only contribution of spans to the presentation of corelations.
\smallskip

It is worth noticing that, even though $\fmod\PID$ is not abelian, the pushout of spans and cospans over $\fmod\PID$ does not have a trivial outcome as for the prop $\F$ of functions (Example~\ref{ex:triv}). Instead, in~\cite{interactinghopf,ZanasiThesis} it is proven that we have the pushout square
\begin{equation}
  \label{eq:pushoutCorelMod2}
  \begin{aligned}
    \xymatrix@C=40pt{
      \fmod{\PID} \oplus \op{\fmod \PID} \ar[d] \ar[r] & \Span{\fmod \PID} \ar[d], \\
      \Cospan{\fmod \PID} \ar[r] &\SV\frPID \pullbackcorner
    }
  \end{aligned}
\end{equation}
in $\PROP$, where $\frPID$ is the \emph{field of fractions} of $\PID$. 

The pushout \eqref{eq:pushoutCorelMod2} is relevant for the categorical
semantics for signal flow graphs pursued in \cite{Bonchi2014b,Bonchi2015,BonchiSZ17}. Even though it is not
an instance of Theorem~\ref{thm:corelationsPROPs} or Corollary \ref{thm:corelationsAbPROP}, our developments shed light on \eqref{eq:pushoutCorelMod2} through the comparison with \eqref{eq:pushoutCorelMod1}. First, note that any element $r \in \PID$ yields a module homomorphism $x \mapsto rx$ in $\fmod \PID$ of type $1 \to 1$, represented as a string diagram $\scalar$. The key observation is that, in \eqref{eq:pushoutCorelMod2}, $\Span{\fmod \PID}$ is contributing to the axiomatisation of $\SV\frPID$ (\emph{cf.} Corollary \ref{cor:presentation}) by adding, for each $r$, an equation 
\begin{equation*}
  \scalar \poi \scalarop \ = \ \IdDiag \poi \IdDiag \text{, corresponding to a pullback}
  \vcenter{\xymatrix@=10pt{
    & 1 & \\
    1\ar[ur]^{\scalar} && 1 \ar[ul]_{\scalar} \\
    & \ar[ul]^{\IdDiag} 1 \ar[ur]_{\IdDiag} &
  }}
\end{equation*}
in $\fmod\PID$. Back to \eqref{eq:pushoutCorelMod1}, the only equations of this kind that $\Span{\mfmod \PID}$ is contributing with are those in which $\scalar$ is a \emph{split mono}, that means, when $r$ is \emph{invertible} in $\PID$. Therefore, the difference between \eqref{eq:pushoutCorelMod1} and \eqref{eq:pushoutCorelMod2} is that in the latter one is adding formal inverses $\scalarop$ also for elements $\scalar$ which are not originally invertible in $\PID$. This explains the need of the field of fractions $\frPID$ of $\PID$ in expanding the pushout object from $\Corel{\fmod\PID}$ (in  \eqref{eq:pushoutCorelMod1}) to $\Corel{\Vect{\frPID}} \cong \SV\frPID$ (in  \eqref{eq:pushoutCorelMod2}).

\subsection{From Maps of Algebras to Relations of Algebras}
\label{ssec.monadalg}
Let $\catC$ be a regular category in which all regular epimorphisms split,
and let $T$ be a monad on $\catC$. Then the
Eilenberg--Moore category $\catC^T$ is regular. As for any regular
category, the (regular epi,mono)-factorisation system is stable, so we can
construct the category $\Rel{\catC^T}$ of relations in $\catC^T$. With a few
further conditions on $T$, we can apply Corollary
\ref{cor.dual} to realise $\Rel{\catC^T}$ as a pushout of categories.

To see this, note that the regular epis in $\catC^T$ are simply the algebra
maps with underlying map in $\catC$ a regular epi. Indeed, since by assumption
regular epimorphisms in $\catC$ split, coequalizers in $\catC^T$ can be
computed using coequalizers in $\catC$. Moreover, as the forgetful functor
$U\maps \catC^T \to \catC$ is monadic, it creates finite limits, and the
canonical map from a span to the pullback of its pushout can also be computed
in $\catC$.  Thus $\catC^T$ satisfies Assumption \ref{ass:threlations} whenever
it is finitely cocomplete and $\catC$ satisfies Assumption
\ref{ass:threlations}. (Given the finite completeness of $\catC$, it is in fact
enough for $\catC^T$ to have reflexive coequalizers: this implies finite
cocompleteness.) This allows us to apply the construction of Corollary
\ref{cor.dual}. 

These conditions are met, for example, for any monad $T$ over $\Vect{\field}$.
Hence, for example, we can apply Corollary \ref{cor.dual} to the construction
of the category of relations between algebras over a field (that is, vector
spaces equipped with a bilinear product).

\section{Concluding remarks} \label{sec:conclusion}

In summary, we have shown that categories of (co)relations may, under certain
general conditions, be constructed as pushouts of categories of spans and
cospans.  In particular, especially since categories of spans and cospans can
frequently be axiomatised using distributive laws, this offers a method of
constructing axiomatisations of categories of (co)relations.  Our results extend
to the setting of props, and more generally symmetric monoidal categories.
Moreover, these results are readily illustrated, unifying a diverse series of
examples drawn from algebraic theories, program semantics, quantum computation,
and control theory.

Looking forward, note that in the monoidal case the resulting (co)relation
category is a so-named hypergraph category: each object is equipped with a
special commutative Frobenius structure.  Hypergraph categories are of
increasing interest for modelling network-style diagrammatic languages, and
recent work, such as that of decorated corelations \cite{Fon16} or the
generalized relations of Marsden and Genovese \cite{MG17}, gives precise methods for
tailoring constructions of these categories towards chosen applications. Our
example on relations in categories of algebras for a monad
(Subsection~\ref{ssec.monadalg}) hints at general methods for showing the
present universal construction applies to these novel examples. We leave this as
an avenue for future work.

\bibliographystyle{plainurl}
\bibliography{calcobib}

\newpage \appendix

\section{Omitted Proofs}

\subsection{Proof of Proposition \ref{lemma:charCospanToCorel}}\label{sec:minorproof}

\begin{proof}[Proof of Proposition \ref{lemma:charCospanToCorel}]
We focus on relations, the proof for corelations being dual. It suffices to show that two spans $\tl{f} \tr{g}$ and $\tl{f'}\tr{g'}$
represent the same relation if and only if the $\Cmono$ parts of the
factorisations of $\langle f,g\rangle$ and $\langle f',g' \rangle$ are
equal.

For the backward direction, write factorisations $\langle f,g\rangle = e;m$ and
$\langle f',g'\rangle = e';m$, and note that $m = \langle m;p_1,m;p_2\rangle$,
where the $p_i$ are the canonical projections. Thus the following diagrams commute
  \[
    \xymatrix@R=5pt{
      &&  \ar[dll]_{f} \ar[ddd]^e \ar[drr]^{g} &&\\
      &&&& \\
      &  \ar[ul]^{p_1}&& \ar[ur]_{p_2} &\\
      &&   \ar[ul]^{m}  \ar[ur]_{m}&&
    } \qquad \qquad
        \xymatrix@R=5pt{
      &&  \ar[dll]_{f'} \ar[ddd]^{e'} \ar[drr]^{g'} &&\\
      &&&& \\
      &  \ar[ul]^{p_1}&&  \ar[ur]_{p_2} &\\
      &&   \ar[ul]^{m}  \ar[ur]_{m}&&
    }
  \]
Therefore $\tr{e \in \Cepi}$ and $\tr{e' \in \Cepi}$ witness that both $\tl{f} \tr{g}$ and
$\tl{f'}\tr{g'}$ are in the equivalence class of
$\tl{p_1}\tl{m}\tr{m}\tr{p_2}$ and so represent the same relation.

For the forward direction, note that if $\tl{f} \tr{g}$ and $\tl{f'}\tr{g'}$
represent the same relation, then there exists a sequence of spans $\tl{f_i}
\tr{g_i}$ in $\catC$ together with morphisms $\tr{e_i \in \Cepi}$,
$i=0, \dots n$, such that $f_1=f$, $g_1=g$, $f_n =f'$, $g_n= g'$, and for all $i
= 1,\dots ,n$ either (i) $e_i;f_i = f_{i-1}$ and $e_i;g_i=g_{i-1}$, or (ii) $f_i
=e_i;f_{i-1}$ and $g_i=e_i;g_{i-1}$. This implies either (i) $e_i;\langle f_i,g_i\rangle
=\langle f_{i-1},g_{i-1}\rangle$ or (ii) $\langle f_i,g_i\rangle
=e_i;\langle f_{i-1},g_{i-1}\rangle$. In either case, by the uniqueness of
factorisations, we see that the $\Cmono$ parts of $\langle f_i,g_i \rangle$ are
the same for all $i$.
\end{proof}

\subsection{Proof of Theorem \ref{thm:corelations}}\label{sec:proof}

We devote this section to give a step-by-step argument for Theorem \ref{thm:corelations}. 

\begin{prop}\label{prop:pushoutcommutes} The square~\eqref{eq:pushoutCorel} commutes. \end{prop}
\begin{proof}
  As ${\subc \+{\subc} \op{\subc}}$ is a pushout, it is enough to show
  \eqref{eq:pushoutCorel} commutes on the two injections of $\subc$, $\op{\subc}$ into ${\subc
  \+{\subc} \op{\subc}}$. This means that we have to show, for any $f \colon
  a\to b$ in $\subc$, that
  \[
    \SpanToCorel(\tl{\id}\tr{f}) = \CospanToCorel(\tr{f}\tl{\id}) \quad\text{ and }\quad
    \SpanToCorel(\tl{f}\tr{\id}) = \CospanToCorel(\tr{\id}\tl{f}).
  \]
  These are symmetric, so it suffices to check one.  This follows immediately
  from the fact that the pushout of $\tl{\id}\tr{f}$ is $\tr{f}\tl{\id}$.
\end{proof}

Suppose we have a cocone over $\Cospan\catC \longleftarrow \subc\+{\subc}\subc
\longrightarrow \Spanc\subc$.  That is, suppose we have the commutative square:
\begin{equation}
\label{eq:arbitrary}
\tag{$\dag$}
\raise15pt\hbox{$
\xymatrix@C=40pt{
{\subc \+{\subc} \op{\subc}} \ar[r] \ar[d] & {\Spanc{\subc}} \ar[d]^{\Psi} \\
{\Cospan{\catC}} \ar[r]_-{\Phi} & {\mathcal{X}}.
}$}
\end{equation}
We prove two lemmas, from which the main theorem follows easily.
\begin{lem} \label{lemma:arbitraryCommutativeDiag} 
  If $\tr{p} \tl{q} $ is a cospan in $\subc$ with pullback $\tl{f} \tr{g}$ (in
  $\catC$), then $\Psi(\tl{f} \tr{g}) = \Phi(\tr{p} \tl{q} )$. Similarly, if
  $ \tl{f} \tr{g} $ is a span in $\subc$ with pushout $\tr{p}\tl{q}$ (in
  $\catC$), then $\Psi(\tl{f} \tr{g} ) = \Phi(\tr{p} \tl{q})$.  
\end{lem}
\begin{proof}
  Consider $\tr{p}\tl{q} \in \subc \+{\subc} \op{\subc}$.  Its image in
  $\mathcal X$ via the lower left corner of the commutative square
  \eqref{eq:arbitrary} is $\Phi(\tr{p}\tl{q})$ while, recalling that
  $\tr{p}\tl{q}  \in \subc \+{\subc} \op{\subc}$ is mapped to $\tl{f}\tr{g}$ in $\Span{\subc}$, its image via the upper right
  corner is $\Psi(\tl{f}\tr{g})$. Thus $\Phi(\tr{p}\tl{q}) = \Psi(\tl{f}\tr{g})$.

  The second claim is analogous, beginning instead with the span
  $\tl{f}\tr{g}$.
\end{proof}

\begin{lem} \label{lemma:functordescendstocorel}
  If $\tr{p_1}\tl{q_1}$ and $\tr{p_2}\tl{q_2}$
  are cospans in $\catC$ such that $\CospanToCorel(\tr{p_1}
  \tl{q_1})=\CospanToCorel(\tr{p_2} \tl{q_2})$, then $\Phi(\tr{p_1} \tl{q_1})
  = \Phi(\tr{p_2} \tl{q_2})$.
\end{lem}
\begin{proof}
Suppose $\CospanToCorel(\tr{p_1}\tl{q_1}) =
\CospanToCorel(\tr{p_2}\tl{q_2})$ as per hypothesis. Then by Proposition \ref{lemma:charCospanToCorel} there exists
$\tr{m_1},\tr{m_2} \in \Cmono$ and $\tr{p}\tl{q} \in \Cospan\catC$ such that
\[
  \tr{p_1}\tl{q_1} \ = \ \tr{p}\tr{m_1}\tl{m_1}\tl{q} \quad\mbox{and}\quad
\tr{p_2}\tl{q_2} \ = \ \tr{p}\tr{m_2}\tl{m_2}\tl{q}.
\]
Then
\begin{eqnarray*}
\Phi(\tr{p_1}\tl{q_1}) &=& \Phi(\tr{p}\tr{m_1}\tl{m_1}\tl{q})
\\ &=& \Phi(\tr{p}\tl{id})\poi\Phi(\tr{m_1}\tl{m_1})\poi\Phi(\tr{id}\tl{q})
\\ &\overset{(\clubsuit)}=&
\Phi(\tr{p}\tl{id})\poi\Psi(\tl{\id}\tr{\id})\poi\Phi(\tr{\id}\tl{q})
 \\  &=& \Phi(\tr{p}\tl{id})\poi\Phi(\tr{\id}\tl{q})
 \\  &=& \Phi(\tr{p}\tl{q}),
\end{eqnarray*}
and similarly for $\tr{p_2}\tl{q_2}$.
The equality $(\clubsuit)$ holds because, by Assumption~\ref{ass:thcorelations}, $\Cmono \subseteq \subc$ and $\tr{m_1 \in \Cmono}$ is mono, thus the pullback of
$\tr{m_1}\tl{m_1}$ is $\tl{\id}\tr{\id}$ and via Lemma~\ref{lemma:arbitraryCommutativeDiag} $\Phi(\tr{m_1}\tl{m_1}) = \Psi(\tl{\id}\tr{\id})$. \end{proof}

\begin{proof}[Proof of Theorem~\ref{thm:corelations}]
Suppose we have a commutative
diagram \eqref{eq:arbitrary}. It suffices to show that there exists a
functor $\theta \colon \Corel{\catC}\to\mathcal{X}$ with
$\theta\CospanToCorel=\Phi$ and $\theta\SpanToCorel =
\Psi$. Uniqueness is automatic by fullness (Proposition~\ref{prop:CospanToCorelFull})
and bijectivity on objects of~$\CospanToCorel$.

Given a corelation $a$, fullness yields a cospan $\tr{f}\tl{g}$ such that $\CospanToCorel(\tr{f}\tl{g}) = a$. We then define $\theta(a) = \Phi(\tr{f}\tl{g})$. This is well-defined by Lemma \ref{lemma:functordescendstocorel}.

For commutativity, clearly $\theta\CospanToCorel = \Phi$.  Moreover, $\theta\SpanToCorel =
\Psi$: given a span $\tl{f}\tr{g}$ in $\Cmono$, let $\tr{p}\tl{q}$ be its
pushout span in $\catC$. Thus by
Lemma~\ref{lemma:arbitraryCommutativeDiag},
\[
  \Psi(\tl{f}\tr{g}) = \Phi(\tr{p}\tl{q})
  =\theta\CospanToCorel(\tr{p}\tl{q})=\theta\SpanToCorel(\tl{f}\tr{g}). \qedhere
\]
\end{proof}

\subsection{Proof of Theorem~\ref{thm:corelationsPROPs}}\label{app:spansPROP}
We first discuss how to put monoidal structures on $\Cospan{\catC}$,
$\Corel{\catC}$, and $\Spanc{\subc}$, and show that $\CospanToCorel$ and
$\SpanToCorel$ are prop morphisms in this case.

  \begin{prop}
    Let $(\catC,\tns)$ be a prop with pullbacks, and let $\subc$ be a sub-prop
    of $\catC$ stable under pullback. If $\tns$ preserves pullbacks in $\subc$,
    then $(\Spanc{\subc},\tns)$ is a prop.
  \end{prop}
  \begin{proof}
    We need to show the map
    \[
      \tns\maps \Spanc{\subc} \times \Spanc{\subc} \longrightarrow \Spanc{\subc}
    \]
    is functorial. That is, given two pairs $(X \leftarrow N \to Y,\:X'
    \leftarrow N' \to Y')$ and $(Y \leftarrow M \to Z,\: Y' \leftarrow M' \to
    Z')$ of spans in $\subc$, we need to show that the composite of their images
    under $\tns$:
    \[
      X\tns X' \longleftarrow (N\tns N') \times_{Y\tns Y'} (M\tns M') \longrightarrow Z \tns Z'
    \]
    is isomorphic to the image under $\tns$ of their composite:
    \[
      X \tns X' \longleftarrow (N\times_YM) \tns (N'\times_{Y'}M') \longrightarrow Z \tns Z'.
    \]
    This is precisely the hypothesis that pullbacks commute with $\oplus$ in
    $\subc$.
  \end{proof}
 
  Note that dualising the above argument with $\subc = \catC$ yields the fact
  that $(\Cospan{\catC},\tns)$ is a prop whenever $\tns$ preseves pushouts.
  Also note that the inclusions $\subc \to \Spanc{\subc}$ and $\op{\subc} \to
  \Spanc{\subc}$ are prop functors.

\begin{prop}
  If $\catC$ is a prop with a costable factorisation system, and $\Cmono$ is
  closed under $\tns$, then $\Corel{\catC}$ is a prop. Moreover, the quotient
  functor 
  \[
    \CospanToCorel\maps \Cospan{\catC} \to \Corel{\catC}
  \]
  is a prop morphism.
\end{prop}
\begin{proof}
  The first task is to show that $\Corel{\catC}$ is indeed a prop. We show that
  $\tns$ induces a monoidal product, which we shall also write $\tns$, on
  $\Corel{\catC}$. Given two corelations $a$ and $b$, with representatives
  $\tr{f}\tl{g}$ and $\tr{h}\tl{k}$ we define their monoidal product $a \tns b$
  to be the corelation represented by the cospan $\tr{f\tns h} \tl{g \tns k}$.
  This is well defined: given $\tr{f'}\tl{g'}$, $\tr{h'}\tl{k'}$ and $m_1,m_2$
  in $\Cmono$ such that $f'= f;m_1$, $g' = g;m_1$, $h'=h;m_2$, $k' = k;m_2$, the
  monoidality of $\tns$ in $\catC$ implies $f'\tns g' = (f\tns g);(m_1 \tns
  m_2)$ and $h'\tns k' = (h\tns k);(m_1 \tns m_2)$. Since $\Cmono$ is closed
  under $\tns$, $m_1 \oplus m_2$ again lies in $\tns$, and the product
  corelation is independent of choice of representatives.
  
  As prop morphisms are strict monoidal functors, to show that $\CospanToCorel$
  is a prop morphism we just need to check $\CospanToCorel(a \tns b)
  =\CospanToCorel a \tns \CospanToCorel b$, where $a$ and $b$ are cospans. This
  follows immediately from the definition: the monoidal product of the
  corelations that two cospans represent is by definition the corelation
  represented by the monoidal product of the two cospans.
\end{proof}

\begin{prop}
  $\SpanToCorel\maps \Spanc{\subc} \to \Corel{\catC}$ is a prop morphism.
\end{prop}
\begin{proof}
  Again, we just need to check that $\SpanToCorel(a\tns b) = \SpanToCorel a \tns
  \SpanToCorel b$. For this we need that the monoidal product preserves pushouts. 
  Indeed, given spans $a = X \tl{f} N \tr{g} Y$ and $b = X' \tl{f'} N'
  \tr{g'} Y'$, we have $\SpanToCorel(a\tns b)$ represented by the cospan 
  \[
    X\tns X' \longrightarrow (X\tns X')+_{(N \tns N')}(Y\tns Y')
    \longleftarrow Y \tns Y',
  \]
  and $\SpanToCorel a \tns \SpanToCorel b$ represented by the cospan
  \[
    X\tns X' \longrightarrow (X+_{N}Y)\tns (X'+_{N'}Y')
    \longleftarrow Y \tns Y'.
  \]
  These cospans are isomorphic by the fact $\tns$ preserves pushouts, and
  hence represent the same corelation.
\end{proof}

Now having described how to interpret \eqref{eq:pushoutCorel} in the category of
props, it remains to show that it is a pushout.

\begin{proof}[Proof of Theorem \ref{thm:corelationsPROPs}]
  From Theorem~\ref{thm:corelations}, we know that \eqref{eq:pushoutCorel}
  commutes, and that given some other cocone
  \[
    \raise15pt\hbox{$
      \xymatrix@C=40pt{
	{\subc \+{\catC} \op{\subc}} \ar[r] \ar[d] & {\Spanc{\subc}} \ar[d]^{\Psi} \\
	{\Cospan{\catC}} \ar[r]_-{\Phi} & {\mathcal{X}},
      }$}
  \]
  there exists a unique functor $\theta$ from $\Corel{\catC}$ to $\mathcal{X}$.
  All we need do here is check that $\theta$ is a prop functor.

  Suppose we have corelations $a$ and $b$, and write $\tilde a$ and $\tilde b$
  for cospans that represent them. Recall that by definition $\theta a = \Phi
  \tilde a$. Then the strict monoidality of $\Phi$ gives
  \[
    \theta(a\tns b) = \Phi(\tilde a \tns \tilde b) = \Phi \tilde a \tns
    \Phi \tilde b = \theta a \tns \theta b.
  \]
  This proves the theorem.
\end{proof}

\begin{proof}[Proof of Corollary \ref{thm:corelationsAbPROP}]
  Note that in an abelian prop the biproduct, being both a product and a
  coproduct, preserves both pushouts and pullbacks, and that the monos are
  closed under the biproduct. Thus we can apply
  Theorem \ref{thm:corelationsPROPs}.
\end{proof}

Finally, note that the above arguments, with the routine care paid to
coherence maps, extend easily to the more general case of symmetric monoidal
categories. In this case the pushout square is a pushout in both the category of
symmetric monoidal categories and lax symmetric monoidal functors, and as well
as the category of symmetric monoidal categories and strict monoidal functors.

\end{document}